\RequirePackage{fix-cm}
\documentclass[smallextended]{svjour3}       
\smartqed  
\usepackage{graphicx}
\usepackage{etex}
\usepackage{paralist}
\usepackage{extpfeil}

\usepackage{tabularx}
\usepackage{arydshln}
\usepackage{hhline}

\usepackage{mathptmx}       
\usepackage{helvet}         
\usepackage{courier}        
\usepackage{type1cm}        
 \usepackage{amsmath}
 \usepackage{mathrsfs}
\usepackage[all]{xy}
\usepackage{bigstrut} 
\usepackage{multirow}
\usepackage{array}
\usepackage{amssymb}

\usepackage{fancyhdr}
\usepackage{float}
\usepackage{pgf,tikz}
\usetikzlibrary{arrows}
\usepackage{amsfonts}
\usepackage{graphicx, wrapfig}
\usepackage{fixltx2e}
\usepackage{makeidx}         
\usepackage{graphicx}        
\usepackage{multicol}        
\usepackage[bottom]{footmisc}%

\usepackage{amsmath,amssymb}
\usepackage{graphicx,graphics,wrapfig}
\usepackage{geometry}
\usepackage{epsf}
\usepackage{latexsym}
\usepackage{epsfig}
\usepackage{lineno}
\usepackage[utf8]{inputenc}
\usepackage[english]{babel}
\usepackage[T1]{fontenc}
\usepackage{amsfonts}
\usepackage{amssymb}
\usepackage{mathrsfs}
\usepackage[round]{natbib}
\usepackage{subfigure}

\DeclareGraphicsExtensions{.pdf,.png,.jpg,.gif}

\usepackage{caption}
\makeatletter
\newcommand{\thickhline}{%
    \noalign {\ifnum 0=`}\fi \hrule height 0.3pt
    \futurelet \reserved@a \@xhline
}
\makeatother

\begin{document}

\title{On the information content of discrete phylogenetic characters }



\author{Magnus Bordewich, Ina Maria Deutschmann, Mareike Fischer, Elisa Kasbohm, Charles Semple, Mike Steel
}

\authorrunning{Bordewich, Deutschmann, Fischer, Kasbohm,  Semple, Steel} 

\institute{Magnus Bordewich:\at
              School of Engineering and Computing Sciences, University of Durham, Science Laboratories, South Road, Durham DH1 3LE  
             \email{m.j.r.bordewich@durham.ac.uk} 
\\
Ina Maria Deutschmann, Mareike Fischer, Elisa Kasbohm:\at
              Institute of Mathematics and Computer Science, Ernst-Moritz-Arndt-University Greifswald, Walther-Rathenau-Str. 47, 17487 Greifswald, Germany 
             \email{email@mareikefischer.de}     \\   
             Charles Semple, Mike Steel: \at
    School of Mathematics and Statistics, University of Canterbury, Private Bag 4800, Christchurch 8140
             \email{charles.semple@canterbury.ac.nz, mike.steel@canterbury.ac.nz, } 
            }

\date{Received: date / Accepted: date}

\maketitle

\begin{abstract} Phylogenetic inference aims to reconstruct the evolutionary relationships of different species based on genetic (or other) data.
Discrete characters are a particular type of data, which  contain information on how the species should be grouped together. However, it has long been known that some characters contain more information than others. For instance, a character that assigns the same state to each species groups all of them together and so provides no insight into the relationships of the species considered. At the other extreme, a character that assigns a different state to each species also conveys no phylogenetic signal. 
 In this manuscript, we study a natural combinatorial measure of the information content of an individual character and analyse properties of characters that provide the maximum phylogenetic information, particularly, the number of states such a character uses and how the different states have to be distributed among the species or taxa of the phylogenetic tree.

\end{abstract}

\section{Introduction}

The evolutionary history of a set of species (or, more generally, taxa) is usually described by a {\em phylogenetic tree}. Such trees can range from small trees on a clade of closely related species, through to large-scale phylogenies across many genera (such as the Tree of Life project \citep{treeoflife}). Phylogenetic trees are usually derived from genetic data, such as aligned DNA, RNA or protein sequences, genetic markers (SINEs, SNPs etc), gene order on chromosomes and  the presence and absence patterns of genes across species.
These types of data generally consist of discrete {\em characters}, each of which  assigns a state from some discrete set to each species. 

In order to derive a tree from character data, we require a measure of how well the characters  `fit' onto each possible tree in order to choose the tree which gives the best fit. One such simple measure is the notion of a character being  homoplasy-free on the tree, which means that the evolution of the character can be explained by assuming that each state has evolved only once.\footnote{This condition is weaker than the assumption that each state actually evolves  only once, since the states at the leaves may have evolved with homoplasy (reversals or convergent evolution) yet still be homoplasy-free on the tree.} It turns out out that this is equivalent to a more combinatorial condition of requiring the character to be `convex'  on the tree. This notion is defined formally in the next section, but, briefly and roughly speaking, it says that when all species (at the  leaves of the tree) that are in the same state are connected to one another,  the resulting subtrees do not intersect. This concept is illustrated in Fig.~\ref{fig:convex}.

\begin{figure}[htb]
\center
\scalebox{.7}{ \includegraphics{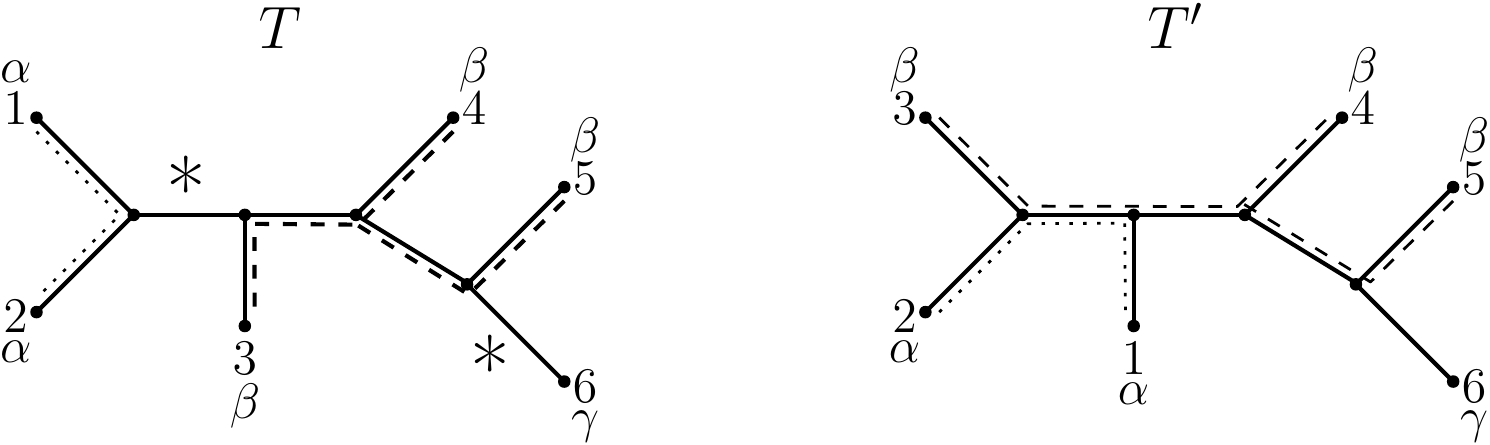} }
\caption{Character $\chi=\alpha\alpha\beta\beta\beta\gamma$ is convex on tree $T$ (left), but not on $T'$ (right). The dotted lines represent the minimum spanning tree connecting the leaves that are in state $\alpha$, whereas the dashed lines represent the minimum spanning tree connecting all leaves which are in state $\beta$. When a character that takes $k\geq 2$ states is convex on a tree, then at least $k-1$ edges are not needed in any of the spanning trees, as  shown for $T$ by the edges marked with an asterisk. }\label{fig:convex}
\end{figure}

In practice, biologists  generally build a tree by using a large number of characters.  However,  it has been shown that for any binary tree $T$ (involving {\em any} number of leaves) just  four characters (on a large enough number of states) suffice to ensure that $T$ is the only tree on which those four characters are convex (\cite{foursuffice}, \cite{bor06}). 
Moreover, even a single character already contains some information concerning which of the species should be grouped together. 
 
Note that a character is often compatible with more than one tree -- for instance, if you have six species (say $1,2, \ldots, 6$),  and the constant character $\chi$ that assigns each species the state $\alpha$, then the induced partition is $\{1,2,3,4,5, 6\}$. This implies that all species are grouped together and therefore no information concerning which species is most closely related to another species can be obtained. This particular  character is convex on all possible phylogenetic trees on six species, so this character does not provide any information on which tree should be chosen.  At the other extreme, a character for which  each species is in a different state from any other species is convex on every possible phylogenetic tree, and so it is also completely uninformative.  The same is true for a character in which some species are in one state, and each remaining species has its own unique state. 

However, if you have the character $\chi$ that assigns Species $1$ and $2$ state $\alpha$, Species $3,4$ and $5$ state $\beta$, and Species $6$ state $\gamma$, then this character is convex on some phylogenetic trees on six taxa, but not on all of them ({\em cf.} Fig. \ref{fig:convex}). 
Under the convexity criterion, such a character would clearly favour some trees over others and thus it contains some information about the trees it will fit on (namely, in this example, all trees that group Species 1 and 2 together versus Species 3, 4 and 5, which will form another group, and Species 6 will form a third group). 
Thus the number of states employed by a character as well as the number of species that are assigned a given state play an important role in deciding how much information is contained in a character. Note that our definition of phylogenetic information is purely combinatorial, and thus differs from some other approaches that are based on particular statistical models (see e.g. \cite{tow07}).

The aim of this paper is to characterize and analyse the characters that have the highest  information content in this sense (i.e. that are convex on relatively few trees and thus have a preference for these few trees over all others), when the number of states is either fixed or free to vary.  Our first main result, Theorem \ref{thm:equalblocksizes}, states that for a fixed number of states, a most informative character will be one in which the subsets (`blocks')  of species in each state are roughly the same size; more specifically, their sizes can only differ by  at most 1.  Moreover, we note that the optimal number of such blocks in a character in order to make it convex on only a few trees cannot easily be determined, as it does not grow uniformly with the number of species because `jumps' appear in the growth function. We analyse these jumps and also provide an approximation without such jumps, and explore the associated asymptotic estimate of the rate of growth (with the number of leaves) of the optimal number of states.

\section{Preliminaries}
\label{sec:prelim}
We now introduce some terminology and notation.
Let $X$ be a finite set of species. Such a set is also often called a set of {\em taxa}.  A \textit{phylogenetic $X$-tree} $T$ is an acyclic connected graph with no vertices of degree 2 in which the leaves are bijectively labelled by the elements of $X$. Such a tree is called {\em binary} if all internal vertices have degree 3. We will restrict our analyses on such trees (for reasons we will explain below) and will therefore in the following refer to phylogenetic trees or just trees for short,  even though we mean binary phylogenetic $X$-trees.

Next, we need to define the type of data we are relating to phylogenetic trees. These data are given as {\em characters}: A function $\chi: X \rightarrow \mathcal{S}$, where $\mathcal{S}$ is a set of  {\em character states}, is called a {\em character}, and if $|\chi(X)|=r$, we say that $\chi$ is an \textit{$r$-state character}. 
 
We may assume without loss of generality that $X=\{1,\ldots,n\}$.  Rather than explicitly writing $\chi(1)=c_1$, $\chi(2)=c_2,\ldots, \chi(n)=c_n$ for some states $c_i \in \mathcal{S}$, we normally write $\chi = c_1 c_2\ldots c_n$. The left-hand side of Fig.~\ref{fig:convex} depicts the character $\chi=\alpha\alpha\beta\beta\beta\gamma$ on six taxa on a tree $T$.  

Note that an $r$-state character $\chi$ on $X$ induces a partition $\pi = \pi(\chi)$ of the set $X$ of taxa into $r$ non-empty and non-overlapping subsets $X_1,\ldots, X_r$ of $X$, which can also be called {\em blocks}. For instance, the character $\chi=\alpha\alpha\beta\beta\beta \gamma$ induces the partition $\pi =\{\{1,2\},\{3,4,5\}, \{6\}\}$ (i.e. the blocks  $X_1=\{1,2\}$, $X_2=\{3,4,5\}$ and $X_3=\{6\}$).  For our purposes, the partition induced by a character is usually more important than the particular character itself. For instance, the characters $\chi_1 = \gamma\gamma\alpha\alpha\alpha\beta$ and $\chi_2 = \beta\beta\gamma\gamma\gamma\alpha$ induce the same partition $\pi =\{\{1,2\},\{3,4,5\}, \{6\}\}$ and are thus considered to be equivalent.

Now that we have defined a structure (namely phylogenetic trees) and the partitions associated with discrete character data, we can introduce a measure of how well these data fit on a tree. A character $\chi$ is called {\em convex} on a phylogenetic tree $T$, if the minimal subtrees connecting taxa that are in the same block do not intersect. This means that if you consider one state and colour the vertices on the paths from each taxon in this state to all other taxa in the same state, and if you repeat this (with different colours) for all other states, there will be no vertex that is assigned more than one colour. An illustration of this idea is given in Fig. \ref{fig:convex}, where the character $\chi=\alpha\alpha\beta\beta\beta\gamma$  is convex on $T$ but not on $T'$. Note that if $\chi$ is convex on $T$ and $|\chi(X)|>1$, this colouration may leave some vertices uncoloured, and it may also assign different colours to the endpoints of certain edges. The deletion of these edges would lead to monochromatic subtrees, all of which are assigned a unique colour (i.e. all leaves in any given subtree are in the same state). This can also be seen by considering tree $T$ from  Fig. \ref{fig:convex}, where the dotted lines refer to the subtree spanning all taxa that are in state $\alpha$ and the dashed lines span the taxa in state $\beta$. If we delete the edges indicated by the asterisks (*) in $T$, all subtrees of $T$ are monochromatic, either dotted or dashed, or an isolated leaf.  Thus a convex character induces a partition of $X$  that can also be derived by deleting some edges of $T$.

Recall that a character can be convex on more than one tree.  Moreover,  whenever a character is convex on a non-binary tree $T$, it is automatically convex on all binary trees which are compatible with this tree (i.e. all binary trees which can be derived from $T$ by resolving vertices of degree greater than three by introducing additional edges). This is illustrated in Fig. \ref{fig:convexityNonBin}, where the tree in the middle is non-binary and there are several ways to add an additional edge in order to make it binary. These additions always lead to trees on which the depicted character is still convex. Therefore, and because binary trees are most relevant in biology (as speciation events are usually considered to split one ancestral lineage into two descending lineages rather than more), we exclusively consider binary trees in the following.

\begin{figure}[htb] 
\center
\scalebox{.6}{ \includegraphics{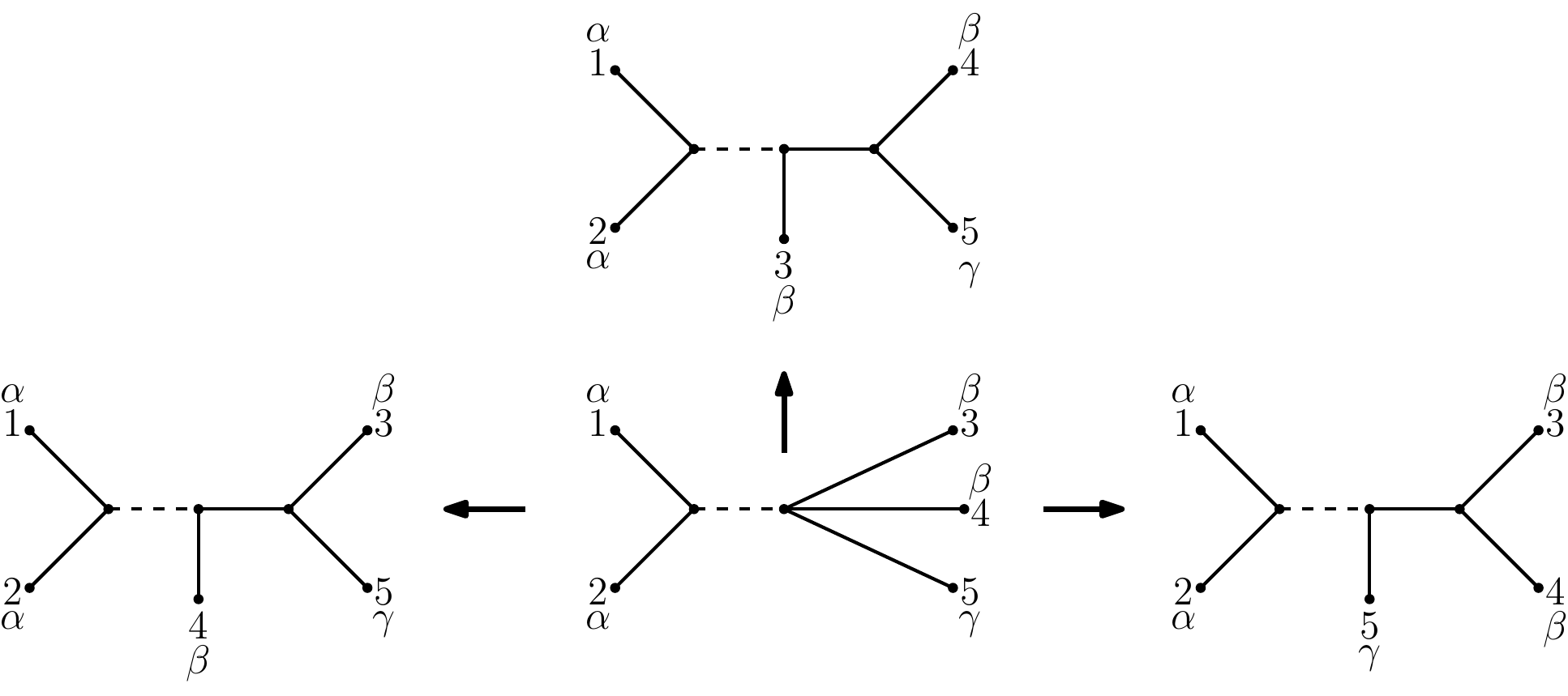} } 
\caption{Character $\chi=\alpha\alpha\beta\beta\gamma$ is convex on the non-binary tree in the middle, but also on all binary trees that are compatible with this tree. The dashed edge is the one that gives rise to the partition $\pi=\{\{1,2\},\{3,4,5\}\}$, which is also induced by $\chi$. Therefore, $\chi$ is convex on all trees which contain this edge.}\label{fig:convexityNonBin}
\end{figure}

Let $b(n)$ denote the number of binary phylogenetic trees on $X=\{1\ldots,n\}$. In total, there are $$b(n)=(2n-5)!! = (2n-5) \cdot (2n-7) \cdots  3  \cdot 1$$ such trees if $n \geq 3$, and $b(1)=b(2)=1$ (see \cite{semplesteel}). As explained above, a character can be convex on more than one tree. However, if a character is convex on all $b(n)$ trees (for some $n \in \mathbb{N}$), it is said to be {\em non-informative}. It is a well-known result that all characters in which at least two states appear at least twice are {\em informative} (see  \cite{bandeltfischer}); in other words, such characters are not convex on all trees, but only on some. As an example, consider again $\chi=\alpha\alpha\beta\beta\beta\gamma$. As explained above and as shown in Fig. \ref{fig:convex}, this character is convex only on some trees, namely those that have an edge separating Species 1 and 2 from Species 3, 4 and 5; and this character uses two of its three character states, namely $\alpha$ and $\beta$,  at least twice (in this case, $\alpha$ is used twice and $\beta$  three times). 

However, the simple distinction between informative and non-informative characters is often not sufficient. In this paper, we want to analyse how much information is contained in an informative character. This can be done by considering the fraction of trees on which the character is convex. Therefore, we denote the number of trees on which a character $\chi$ with induced partition $\pi$ is convex by $N_\pi$, and the fraction of such trees by $P_\pi = \frac{N_\pi}{b(n)}$. 

Note that for a given $r$-state character $\chi$ on $X=\{1,\ldots,n\}$ with the induced partition $\pi=\{X_1,\ldots,X_r\}$, the number $N_\pi$ can be explicitly calculated with the following formula, which was first stated in  \cite[Theorem 2]{carter}: 

\begin{equation}
N_\pi= \frac{b(n)}{b(n-r+2)} \cdot \prod\limits_{i=1}^r b(x_i +1),\label{Nformula}
\end{equation}
where $x_i=|X_i|$ for all $i=1,\ldots,r$ and $b(n)$ denotes (as stated above) the number of binary phylogenetic trees on $X=\{1,\ldots,n\}$.

We  are particularly interested in characters that {\em minimize} $P_\pi$, because they are only convex on the smallest number of trees and therefore contain the most information on which tree they fit `best' (based on the convexity criterion).  Thus, following   \cite{steelpenny}, we define the {\em information content} of a character $\chi$ with induced partition $\pi$ as follows:

\begin{equation}  I_\pi = -\ln \hspace{0.1cm}P_\pi = -\ln \left(\frac{N_\pi}{b(n)}\right).\label{inf}\end{equation}  

Note that searching for a character with minimal $P_\pi$ (i.e. a minimal fraction of trees on which it is convex), is equivalent to searching for a character with maximal $I_\pi$ (i.e. a character with maximal information content). 
Notice also that, by Eqn.~(\ref{Nformula}),  we can write
$I_\pi = \ln (b(n-r+2)) -\sum_{i=1}^r \ln(b(x_i+1))$, and since $b(k)$ is a product of consecutive odd natural numbers, we can further write $I_\pi$ as a sum of the form
$ \sum_{j  \in S}a_j \ln j$, where $S$ is a finite set of odd natural numbers and $a_j$ is an integer for each $j \in S$.
We are now in the position to state our results concerning characters for which $I_\pi$ is maximal.

\section{Results}
\label{sec:3}

\subsection{Maximizing $I_\pi$}\label{subsec:max}
We now investigate the character partitions $\pi$  of a set $X$ of size $n$  that maximize $I_\pi$.  Consider an $r$-state character $\chi$ with the induced partition $\pi = \{X_1,\ldots,X_r\}$ and let $x_i=|X_i|$ (for all $i=1,\ldots,r$) denote the block sizes. The main problem considered in this manuscript, namely maximizing $I_\pi$ (or, equivalently, minimizing $P_\pi$), consists of two combined problems, namely finding the optimal number $r$ of states (i.e. the optimal number of blocks in $\pi$), as well as the optimal block sizes $x_i$ for $i=1,\ldots,r$ (i.e. the distribution of states on taxon set $X$).  

We first consider the latter problem for the case when $n$ and $r$ are fixed.
Let $n\geq 3$ and $r\leq n$ be natural numbers. Let $N(n,r)$ denote the minimum value of $N_\pi$ over all partitions $\pi$ of $X=\{1, \ldots, n\}$ into $r$ blocks.
Formally stated:  $$N(n,r) = \min\limits_{\stackrel{\pi=\{X_1,\ldots,X_r\}:}{ |X_1|+\cdots+|X_r|=n}} N_\pi.$$

Let $l=l(n,r)= r\cdot \lceil \frac{n}{r}\rceil -n$. It is easily shown that:
$$l \left \lfloor \frac{n}{r} \right \rfloor  + (r-l) \left\lceil \frac{n}{r} \right\rceil= n,$$
and so $\{1, \ldots, n\}$ can be partitioned into $l$ sets of size $\left \lfloor \frac{n}{r} \right \rfloor$ and $r-l$ sets of size $\left\lceil \frac{n}{r} \right\rceil$.
The main result of this section is the following.

\begin{theorem} \label{thm:equalblocksizes} 
For $n\geq 3$ and $r\leq n$:
$$N(n,r) = \frac{b(n)}{b(n-r+2)} \cdot b\left(\left\lfloor\frac{n}{r}\right\rfloor+1\right)^{l} \cdot b\left(\left \lceil \frac{n}{r}\right\rceil+1\right)^{r-l},$$
where $l=r\cdot \lceil \frac{n}{r}\rceil -n$.
\end{theorem}

\par\vspace{0.5cm}
\begin{remark} Note that in the case where $r$ is a divisor of $n$, the equation stated in Theorem \ref{thm:equalblocksizes} reduces to $N(n,r) = \frac{b(n)}{b(n-r+2)} \cdot b(\frac{n}{r}+1)^r$, since $\lceil \frac{n}{r}\rceil=\frac{n}{r}$ and thus $l=0$.  \label{rem:equalblocksizes}\end{remark}

The proof of Theorem \ref{thm:equalblocksizes} requires the following technical lemma, which is proved in the Appendix.

\begin{lemma} \label{lem:equalblocksizes} Let $m, s \in \mathbb{N}$, $m\geq 2$ and $s \geq 2$. We then have: $$b(m+s) \cdot b(m) >b(m+s-1)\cdot b(m+1).$$
\end{lemma}

Lemma \ref{thm:equalblocksizes} immediately leads to the following corollary (also derived in \cite{schutz}).

\begin{corollary} \label{cor:equalblocksizes} If a character $\chi$ with induced partition $\pi = \{X_1,\ldots,X_r\}$ and block sizes $x_1,\ldots,x_r$ maximizes $I_\pi$, then for $x_i$ and $x_j$ ($i,j \in \{1,\ldots,r\}$, $i \neq j$), we have: $|x_i-x_j|\leq 1$ (i.e. the block sizes differ by at most 1).
\end{corollary}

\begin{proof} Let $\chi$ be a character with the induced partition $\pi=\{X_1,\ldots,X_r\}$ that maximizes $I_\pi$ (equivalently, which  minimizes $N_\pi$).  Let $x_i=|X_i|$ for all $i=1,\ldots,r$. Assume that there exist $i,j \in \{1,\ldots,r\}$ such that $|x_i-x_j|\geq 2$. Without loss of generality, assume that $x_i>x_j$. Set $m=x_j+1$ and $s=x_i-x_j$. Both $m$ and $s$ are then at least 2 (because $x_j \geq 1$ by definition of partition $\pi$ and $x_i-x_j\geq 2$ by assumption). We apply Lemma~\ref{thm:equalblocksizes} and find that $$b(x_i+1)\cdot b(x_j+1)=b(m+s)\cdot b(m) > b(m+s-1)\cdot b(m+1) = b(x_i) \cdot b(x_j+2).$$ Note that the contribution of $X_i$ and $X_j$ to  $\prod\limits_{i=1}^r b(x_i +1)$ in $N_\pi$ of Eqn. (\ref{Nformula}) is $b(x_i+1)\cdot b(x_j+1)$. However, if we now modify $\chi$ so that we remove one element of $X_i$ and add it to $X_j$, the contribution of this modified character is $b(x_i) \cdot b(x_j+2)$, which we have shown to be smaller than the original contribution. This is a contradiction, as $\chi$ was chosen as a minimizer of $N_\pi$. Therefore, the assumption $|x_i-x_j|\geq 2$ was wrong and thus we have $|x_i-x_j|\leq 1$. This completes the proof. \qed
\end{proof}

\noindent We now use Lemma \ref{lem:equalblocksizes} and  Corollary \ref{cor:equalblocksizes} to prove Theorem \ref{thm:equalblocksizes}.

\begin{proof}[Theorem \ref{thm:equalblocksizes}] 


\noindent Using Eqn. (\ref{Nformula}), the only thing that remains to be shown is that: $$\prod\limits_{i=1}^r b(x_i +1) = b\left(\left \lfloor \frac{n}{r}\right \rfloor+1\right)^{l} \cdot b\left(\left \lceil \frac{n}{r}\right \rceil+1\right)^{r-l}.$$ Considering Remark \ref{rem:equalblocksizes}, we do this by investigating the cases $r\mid n$ and $r \nmid n$ separately. 


\begin{enumerate}
\item Let $r \mid n$ (i.e. $n=k\cdot r$ for some $k \in \mathbb{N}$). Let $\chi$ be a character with induced partition $\pi=X_1,\ldots,X_r$ such that $N_\pi=N(n,r)$ (i.e. $\pi$ minimizes $N_\pi$ for given values of $n$ and $r$). Now assume that not all block sizes are equal to $\frac{n}{r}=k$. There is then an $i \in \{1,\ldots,r\}$ such that $x_i \neq k$. If $x_i > k$, then as $x_1+\ldots + x_r=n$, there must be a $j\in \{1,\ldots,r\}$ such that $x_j <k$ (or vice versa). Let us assume, without loss of generality, that $x_i=k+\hat{s}$ and $x_j=k-\tilde{s}$ for $\hat{s},\tilde{s} \in \mathbb{N}$; in particular,  $\hat{s},\tilde{s}\geq 1$. Then $x_i-x_j=\hat{s}+\tilde{s}\geq 2$. This is a contradiction because, by Corollary \ref{cor:equalblocksizes}, $x_i$ and $x_j$ can differ by at most 1 as $\chi$ minimizes $N_\pi$. Thus  in the case where $n=k\cdot r$, we have $x_i=k=\frac{n}{r}$ for all $i=1,\ldots, r$ and therefore $\prod\limits_{i=1}^r b(x_i +1) =  b(\frac{n}{r}+1)^r$.

\item Next, consider the case where $r \nmid n$. Using Corollary \ref{cor:equalblocksizes}, a character $\chi$ with the  induced partition $\pi=\{X_1,\ldots,X_r\}$ which minimizes $N_\pi$ can only lead to sets of sizes $x_i$, $x_j$, which differ by at most 1. As we need $r$ such sets in total, the only way to achieve this is by allowing $l$ sets of size $\lfloor \frac{n}{r}\rfloor$ and $r-l$ sets of size $\lceil\frac{n}{r} \rceil$ for some $l \in \mathbb{N}$, $l\leq r$ (note that $\lceil \frac{n}{r} \rceil-\lfloor\frac{n}{r} \rfloor=1$ as $r \nmid n$). This has a unique solution, as $n = l \cdot \lfloor \frac{n}{r}\rfloor + (r-l) \cdot \lceil\frac{n}{r} \rceil$ leads to $l=r\cdot \lceil \frac{n}{r}\rceil -n$.  Moreover, this leads to $\prod\limits_{i=1}^r b(x_i +1) = b(\lfloor \frac{n}{r}\rfloor+1)^{l} \cdot b(\lceil \frac{n}{r}\rceil+1)^{r-l}$, which, together with Eqn. (\ref{Nformula}), completes the proof. 
\end{enumerate}
\qed
\end{proof}

\subsection{The number of states ($r_n$) that maximizes $I_\pi$}

As we have seen in Corollary \ref{cor:equalblocksizes} and in the proof of Theorem \ref{thm:equalblocksizes}, a character which has maximal information content $I_\pi$ induces a partition $\pi=\{X_1,\ldots,X_r\}$ of roughly equal block sizes $x_1,\ldots,x_r$. In the  case where $r$ divides $n$, all block sizes are equal to  $\frac{n}{r}$; otherwise, there are $l=r\cdot \lceil \frac{n}{r}\rceil -n$ blocks of size $\lfloor \frac{n}{r}\rfloor$, and all other $r-l$ sets have size $\lceil \frac{n}{r}\rceil$. 

Recall that in order to find characters that maximize $I_\pi$ and thus minimize $N_\pi$, we have to solve two problems: we have to find the optimal value of $r$ as well as the corresponding block sizes $x_i$.

Let
$$I(n,r) = -\ln\left(\frac{N(n,r)}{b(n)}\right),$$
which is the maximal value of $I_\pi$ over all partitions of $\{1, \ldots, n\}$ into $r$ blocks.
Let $r_n$ be the value of $r$ that maximizes $I(n,r)$.

Consider the special case where $n$ is a multiple of $r$. In this case, we  know that the block sizes that maximize $I_\pi$ are exactly $\frac{n}{r}$.   If we only look at this fixed distribution of states, the two problems stated above -- namely finding the optimal value of $r$ and the optimal block sizes $x_i$ -- reduces to just the first problem, namely finding the optimal value of $r$.

\noindent Note that when $r=1$, we have $k=n$ and $|X|=n=x_1$, and thus by Eqn. (\ref{Nformula}) we get:  $$N_\pi = \frac{b(n)}{b(n-1+2)} \cdot b(x_1 +1) = \frac{b(n)}{b(n+1)} \cdot b(n +1) = b(n).$$

\noindent In other words, in the case where a character $\chi$ only employs one character state (say $\alpha$) the resulting character $\chi=\alpha\alpha\ldots \alpha$ on $X=\{1,\ldots,n\}$ is convex on {\em all} $b(n)$ trees on the taxon set $X$, which means that $N_\pi$ is maximal and therefore $I_\pi=-\ln\frac{N_\pi}{b(n)}=-\ln\frac{b(n)}{b(n)}=0$, which is minimal. Similarly, if there are $|X|=n$ different character states employed by $\chi$ (i.e. if $x_i=k=1$ for all $i=1\ldots,r$)  we get:
$$N_\pi = \frac{b(n)}{b(n-n+2)} \cdot \prod\limits_{i=1}^r b(x_i +1) = \frac{b(n)}{b(2)} \cdot \prod\limits_{i=1}^r b(1 +1) = b(n) \cdot b(2)^{r-1} =b(n).$$

\noindent Here, the last two equations use the fact that $b(2)=1$. In particular, if  a character employs $r=n$ character states, this character is also convex on all trees on taxon set $X=\{1,\ldots,n\}$, and thus $I_\pi=0$. 

\par\vspace{0.5cm}
\noindent Therefore, if we wish to minimize $N_\pi$ and thus $P_\pi$ in order to maximize $I_\pi$, the number $r_n$ of character states must lie strictly between $1$ and $n$; otherwise, $N_\pi$ is maximal. Between these boundary cases, it is not obvious how to find $r_n$.  For example, if we fix $n=120$ and exhaustively examine all possible values for $r$ between 1 and $n$, then we find that $r_n=24$. This scenario is depicted in the left-hand portion of Fig.~\ref{case120}.

\noindent Similarly, we randomly sampled values of $n$ between 10 and 10000, and considered just the divisors for each value of $n$ in order to estimate the divisor $r$ of $n$ that maximizes $I_\pi$, where $\pi$ is a partition into $r$ blocks.  The results are depicted in the right-hand portion of Fig. ~\ref{case120}. However, note that we discarded $n$ whenever our random choice of $n$ was a prime number, because then it is clear that the only divisors are $1$ and $n$, which leads to the cases we analysed above for which we know that $N_\pi=b(n)$ and thus $P_\pi = 1$ and so $I_\pi = 0$.

\begin{figure}[htb]
\center
\scalebox{.35}{ \includegraphics{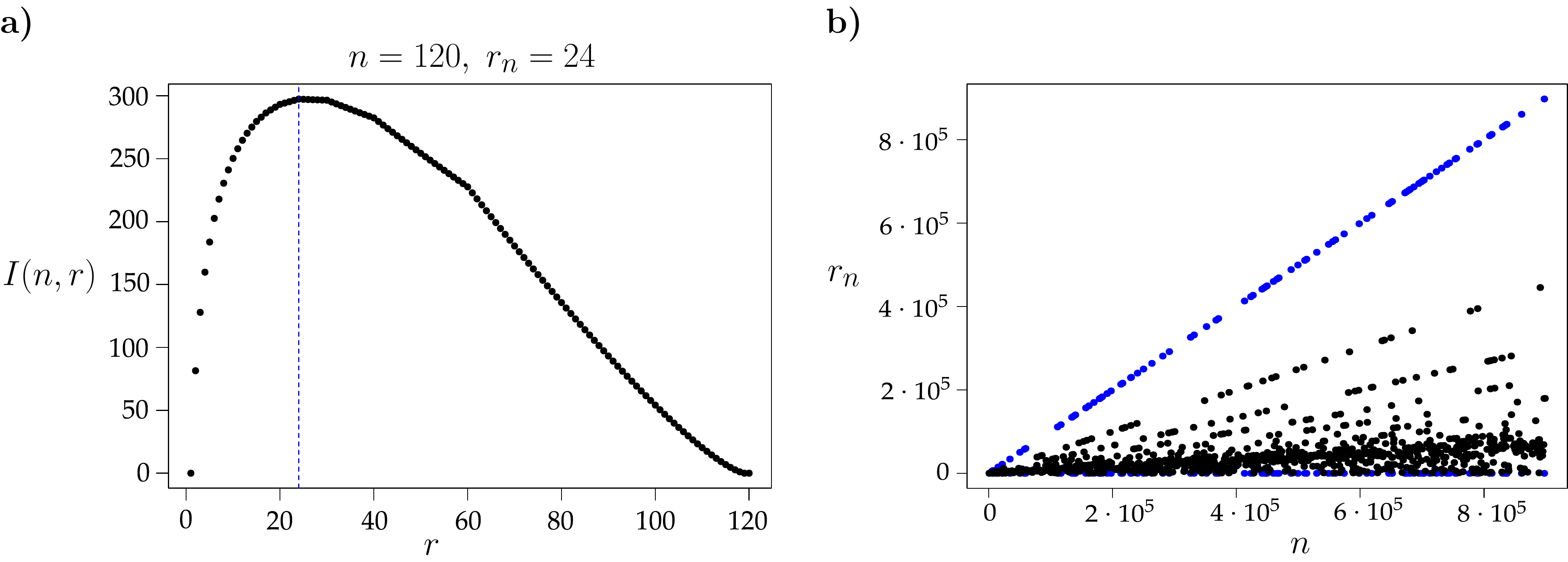} }
\caption{ On the left-hand side, the case $n=120$ is depicted, along with all values of $r$ from 1 to $n$. It can be seen that $I_\pi=-\ln(P_\pi)$ is maximal when $r=24$ is chosen. On the right hand side, the plot shows the divisor $r$ of $n$ for which $I_\pi$ is maximal (where $\pi$ is a partition into $r$ blocks) for randomly chosen values of $n$ between 10 and $10^6$.  The primes in this interval do not allow for any other equal block sizes than one block of size $n$ or $n$ blocks of size 1 (which have equal $I_\pi$ value of 0);  the top (blue) line of dots shows this value $r_n=n$ for the latter choice.}\label{case120}
\end{figure}

\par\vspace{0.5cm}

\par\vspace{0.5cm}
\subsection{Analysis of the growth of $r_n$}\label{subsec:jumps}

\subsubsection{The shape of $I_\pi$ and its consequences for $r_n$}\label{calculations}

By exploiting Theorem \ref{thm:equalblocksizes}, exhaustive searches for $r_n$, given $n$, can be done more efficiently. This is because for each value of $r$, we now know the optimal block sizes, so we do not have to look at all possible partitions. Consequently, an exhaustive search for $r_n$ by testing all possible values of $r$ for a fixed value of $n$ is easily possible up to $n=10000$ (and probably even higher than that).

In order to understand the growth of $r_n$, we first explicitly searched for $r_n$ for each value of $n$ between 1 and 360 ({\em cf.} Fig. \ref{jumps1})) and between 1 and 10000 ({\em cf.} Table \ref{Tab1})). Although Fig. \ref{jumps1} shows that $r_n$ has an increasing trend as $n$ grows as well as piecewise linear growth, there are {\em jumps} back to a smaller number of blocks from time to time. Clearly, the growth of $r_n$ is not uniform. It seems as if the size of the intervals between the jumps increases roughly threefold. Table \ref{Tab1} gives the exact numbers for the jumps for $n \leq 10000$. Note that not only does the distance between the jumps increase, but also the size of the jumps $r_n-r_{n+1}$. However, if we consider the size of the jumps relative to $r_n$, then the jump sizes actually decrease. The sequence of jumps $(9, 30, 104, 345, \ldots)$ does not follow any obvious pattern and could not be matched to any known series of numbers in the On-Line Encyclopedia of Integer Sequences (\cite{sloane2003line}).
\begin{table}[b!] \centering 
\begin{tabular}{|  rrcrr | } 
\hline $n$ & $r_n$ &  $\left \lfloor{\frac{n}{r}}\right \rfloor$ & $\left \lceil{\frac{n}{r}}\right \rceil$ & $- \ln P_\pi$ \\ 
\hline $8$ & $4$ &  $2$ &  $2$ & $4.654$ \\ 
$9$ & $3$ &  $3$ &  $3$ & $5.953$ \\ \hdashline
$29$ & $9$ &  $3$ &  $4$ & $41.016$ \\ 
$30$ & $8$ &  $3$ &  $4$ & $43.151$ \\ \hdashline
$103$ & $25$  & $4$ &  $5$ & $242.696$ \\ 
$104$ & $21$  & $4$ &  $5$ & $245.854$ \\ \hdashline
$344$ & $68$  & $5$ &  $6$ & $1141.630$ \\ 
$345$ & $58$ & $5$ &  $6$ & $1145.770$ \\ \hdashline
$1108$ & $184$  & $6$  & $7$ & $4756.330$ \\ 
$1109$ & $159$  & $6$  & $7$ & $4761.460$ \\ \hdashline
$3484$ & $497$  & $7$ &  $8$ & $18376.200$ \\ 
$3485$ & $436$ & $7$ & $8$ & $18382.300$ \\ 
\hline 
\end{tabular} 
  \caption{All jumps of $r_n$ for $n \leq 10000$.} 
  \label{Tab1} 
\end{table}

\begin{figure}[htb] 
	\center
	\includegraphics[scale=0.43]{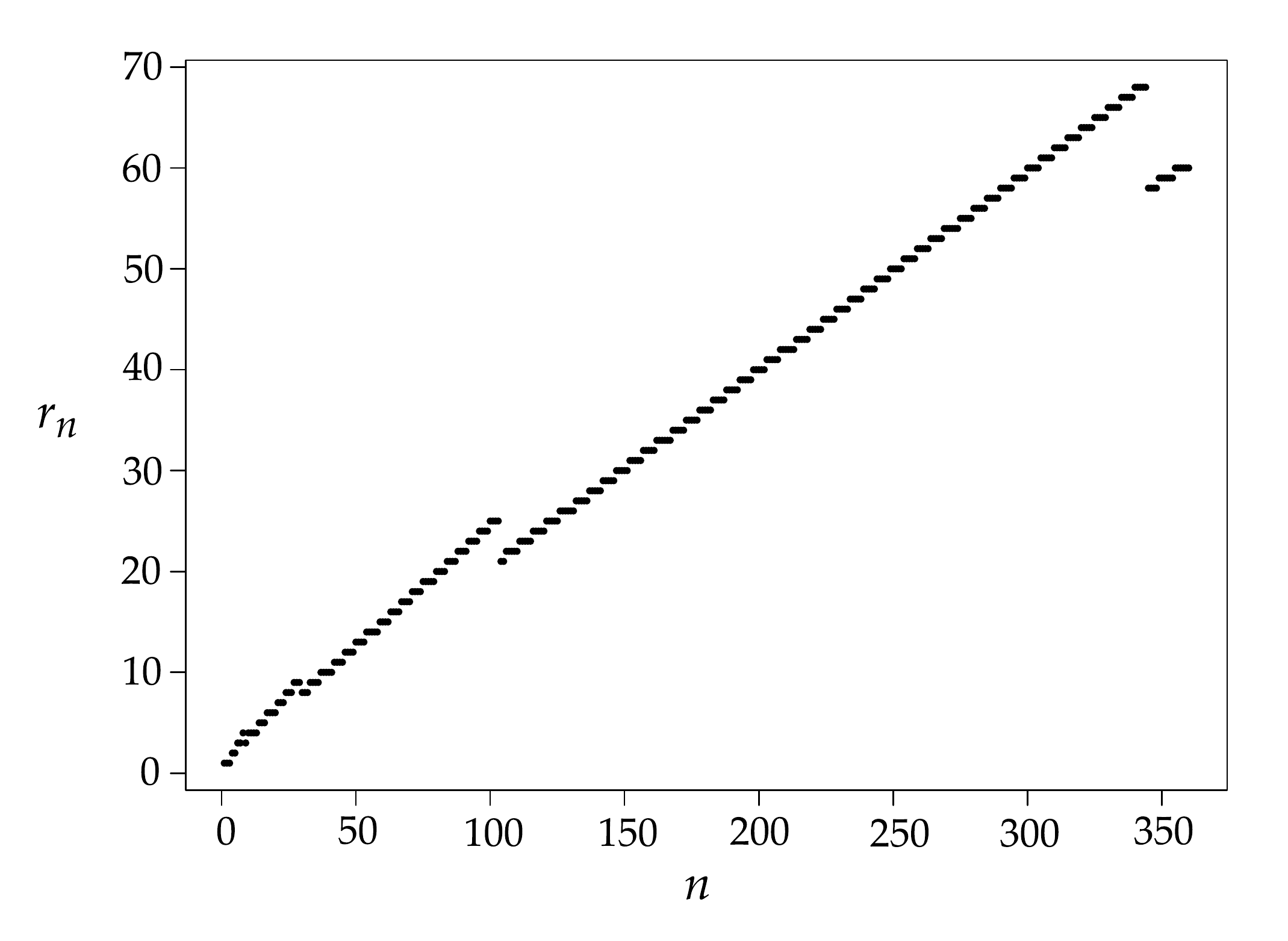}
	\caption[caption of drops to 360]{The values of $r_n$ for values of $n$ between 1 and 360. Note that  $r_n$ drops down at $n=9$ (from $r_n=4$ at $n=8$ to $r_{n+1}=3$), as well as at $n=30$, $n=104$ and $n=345$, as can also be seen in Table \ref{Tab1}.
	\label{jumps1}
	}
\end{figure}

\begin{figure}[htb]
	\begin{centering}
		\scalebox{0.6}{\includegraphics[width=\textwidth]{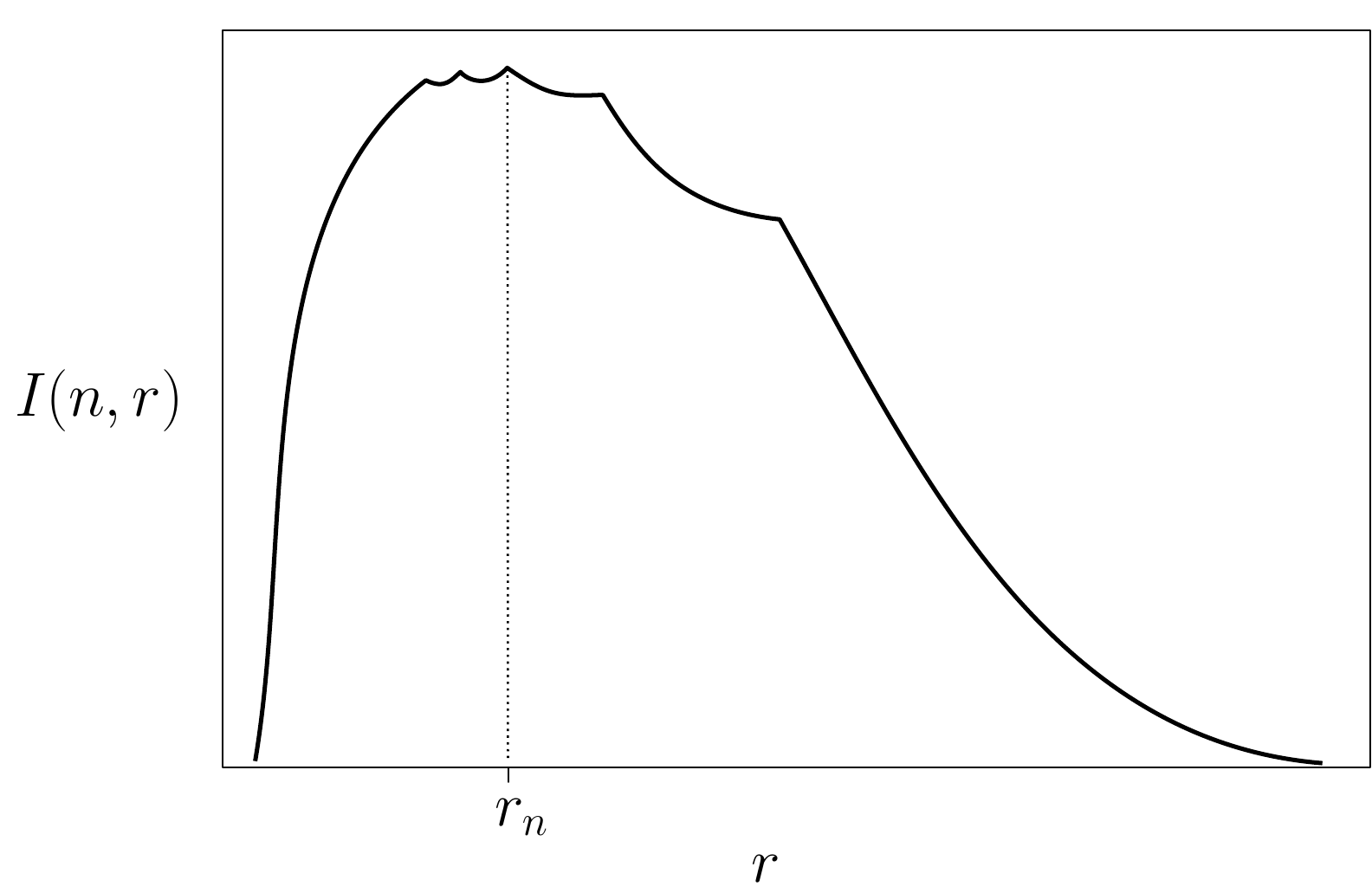}}
		\caption{\label{convP}A simplified sketch of the shape of $I_\pi$ as presented in Fig. \ref{case120}.}\end{centering}
\end{figure}

\begin{figure}[htb]
	\center
	\includegraphics[scale=0.5]{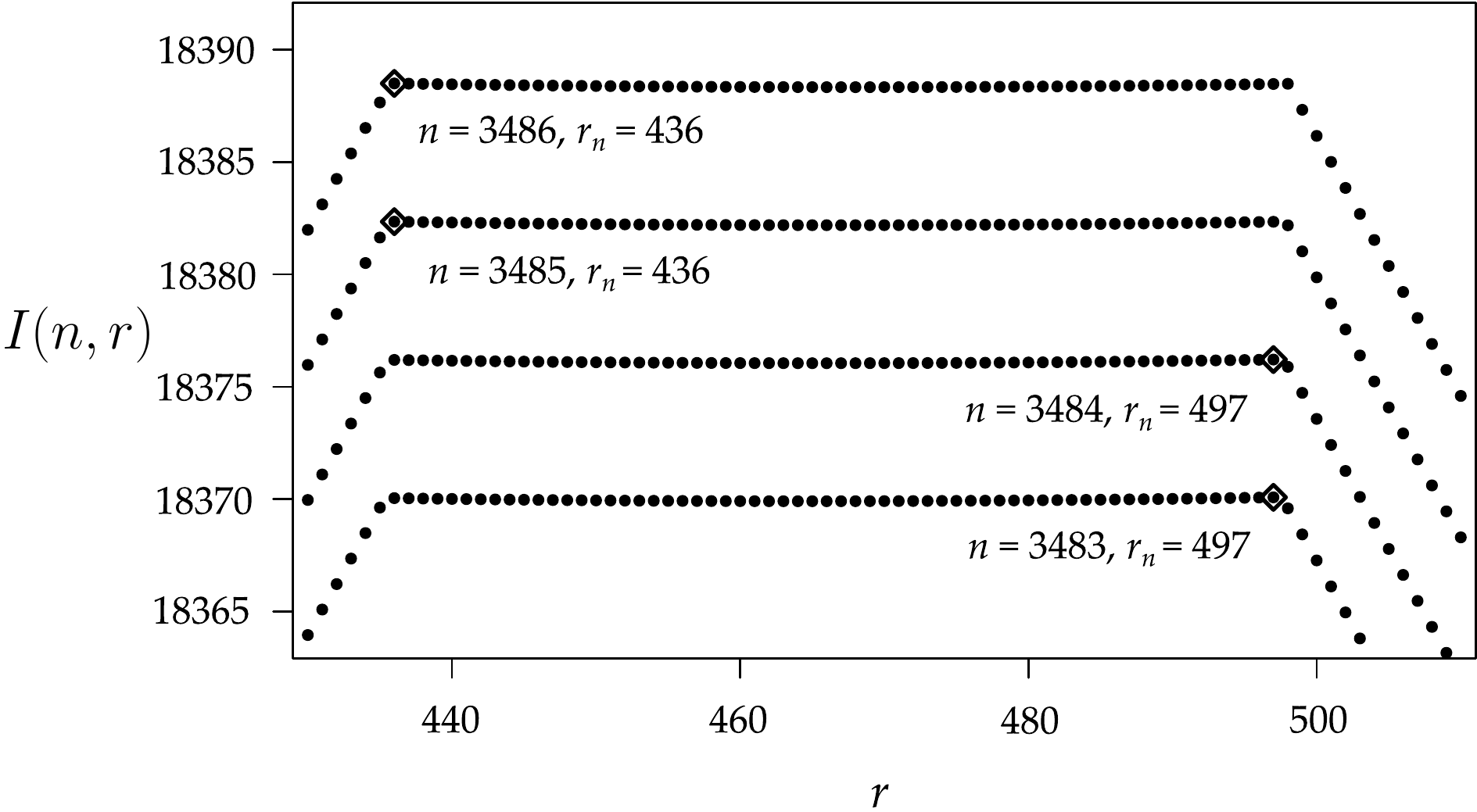}
	\caption{The value $r_n$ jumps between $n=3484$ and $n=3485$ as the maximum switches from the right edge of the convex section to the left edge.}
	\label{flatTops}
\end{figure}

\subsubsection{The shape of $I(n,r)$}\label{derivative}

We now investigate the shape of the function $I(n,r)$ as $r$ increases. For a fixed value of $n$, a closer look at the graph of $I(n,r)$ reveals the reason for the jumps in the block sizes; namely, that the graph is not as smooth as it may seem at first glance. It is instead a concatenation of several convex functions. This can already be guessed from Fig. \ref{case120} (left-hand graph), but in order to make it a bit more obvious, we sketched the plot again (enhancing the shape) in Fig. \ref{convP}. Note that the value of $r_n$ jumps when the maximum of $I(n,r)$ shifts from one edge of a convex section to the other. Fig. \ref{flatTops} shows an example for such a shift at $n=3485$. Here, $r_n$ drops down from $497$ to $436$. This means that the optimal partition for $n=3485$ contains $61$ fewer blocks than the optimal partition for $n=3484$. As can be seen in Fig.~\ref{flatTops}, the jump in $r_n$ is accompanied by a shift of the maximum from being on the right-hand side of a convex segment being on the left-hand side of the next convex segment. 

Table \ref{Tab1} describes  the values of $r$ at which downward jumps in the value of $r_n$ occur.  Before the jump, most of the subsets in an optimal partition $\pi$ are of size $\left \lfloor{\frac{n}{r}}\right \rfloor$, whereas after adding one additional leaf, the optimal partition contains mostly subsets of size $\left \lceil{\frac{n}{r}}\right \rceil$. As $r_n$ does not grow linearly, the block sizes $\left \lfloor{\frac{n}{r_n}}\right \rfloor$ and $\left \lceil{\frac{n}{r_n}}\right \rceil$ do not grow linearly either.  But contrary to $r_n$, the block sizes only alternate by $\pm 1$.\par \vspace{0.5cm}

\subsection{Approximating the rate of growth of $r_n$ with $n$}

In this section, recall the notation $\sim$ for asymptotic equivalence, in which $f(n)\sim g(n)$ is shorthand for
$\lim_{n\rightarrow \infty} f(n)/g(n) = 1$.
We want to investigate the growth of $r_n$ as  $n$ grows. 
Therefore, we need a differentiable approximation of $I_\pi$, as $I_\pi$ is not differentiable (its shape consists of piecewise-convex segments). 
From Theorem \ref{thm:equalblocksizes} we have: 
\begin{equation}
\label{gammago}
I(n,r) = - \ln \left(\frac{N(n,r)}{b(n)}\right) = -\ln \left(\frac{b(\lfloor \frac{n}{r}\rfloor+1)^{l} \cdot b(\lceil \frac{n}{r}\rceil+1)^{r-l}}{b(n-r+2)}\right). 
\end{equation}

Now $b(n+1) \sim \gamma(n)$ for the real-valued function $\gamma$ defined for $x>0$ by
 $\gamma(x)=\frac{1}{\sqrt{2}} \left(\frac{2}{e}\right)^x x^{x-1}$ ({\em cf.} \cite{mcdiarmid}). Let $I_\gamma(n,r)$ denote the approximation to
$I(n,r)$ obtained by first approximating $\lceil \frac{n}{r}\rceil$ and $\lfloor \frac{n}{r}\rfloor$
by $n/r$ (these approximations assume that $n/r\gg 1$), and then using $\gamma(x)$ in place of  $b(x+1)$ in the resulting expression for $I(n,r)$.
Making these substitutions, the expression on the far right of Eqn.~(\ref{gammago}) becomes independent of $l$ and we can write:
\begin{displaymath}
	I_\gamma(n,r)= -\ln \left(\frac{\gamma\left(\frac{n}{r}\right)^r}{\gamma(n-r+1)}\right) = -r \ln\left(\gamma\left(\frac{n}{r}\right)\right) + \ln(\gamma(n-r+1)).
\end{displaymath}

Let $\tilde{r}_n$ denote a  value of $r$ that  maximizes $I_\gamma(n,r)$. We want to use $\tilde{r}_n$ as an estimator for $r_n$. Fig. \ref{rApprox} shows the values of $\tilde{r}_n$ in comparison to $r_n$ as $n$ ranges from $1$ to $1000$ (over this range there is a unique value for $r$
that maximizes $I_\gamma(n,r)$).  Here, it can be seen that $\tilde{r}_n$ gives a reasonable approximation to $r_n$ over the range shown (note that $I_\gamma(n,r)$ deviates from $I(n, r)$ for values of $r$ close to $n$, however in this region $I(n,r)$ is far from its maximal value). 
\begin{figure}[htb]
	\center
	\includegraphics[scale=0.6]{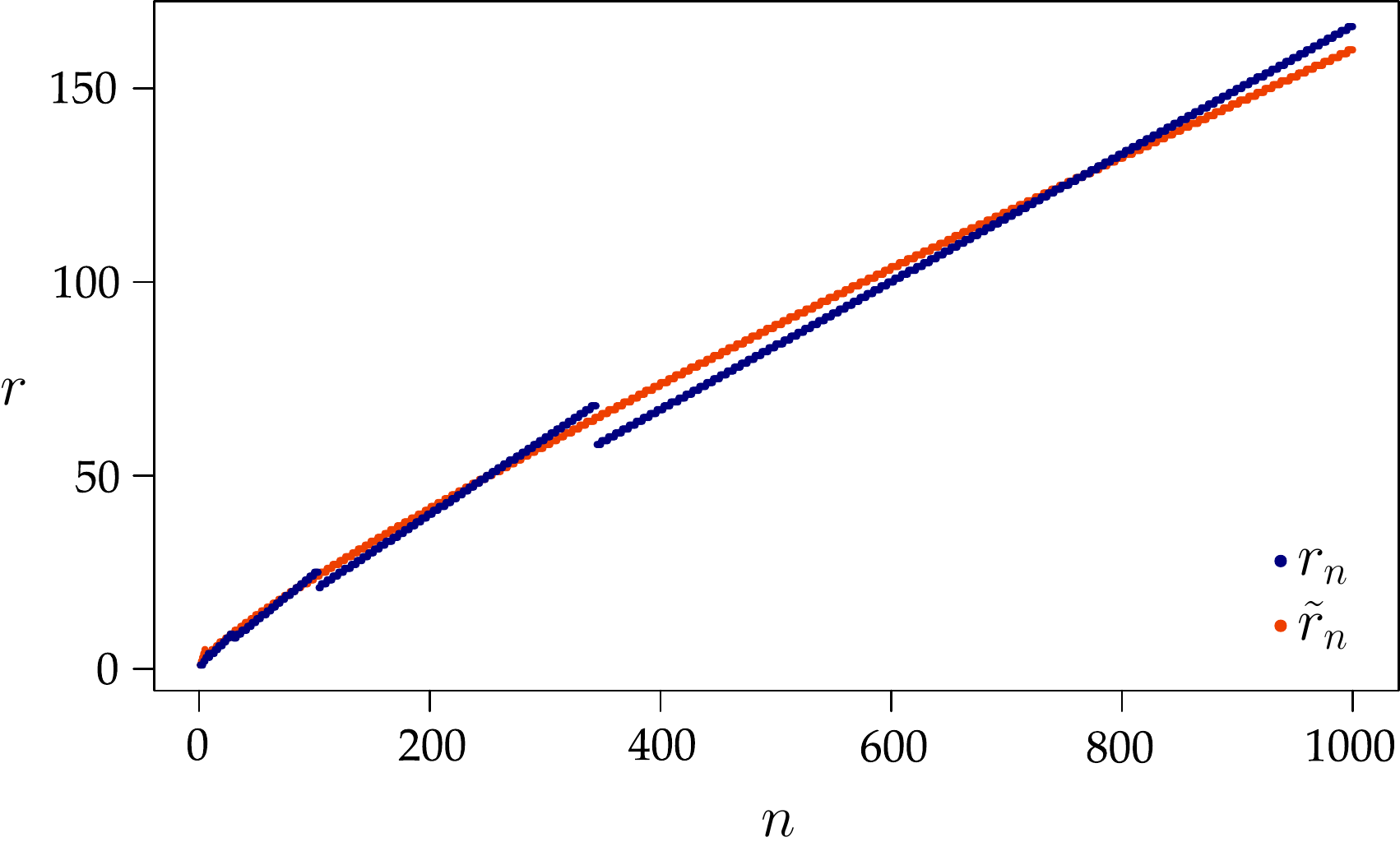}
	\caption{A comparison of $r_n$ (broken curve segements) and $\tilde{r}_n$ (continuous curve) for $n$ from $1$ to $1000$. It can be seen that as opposed to $r_n$, $\tilde{r}_n$ does not have any jumps back to a smaller value, but is instead increasing uniformly.}
	\label{rApprox}
\end{figure}

\begin{theorem}
\label{asythm}
The value(s) of $r=\tilde{r}_n$ at which $I_\gamma(n,r)$ achieves its maximum value satisfies the asymptotic equivalence $\tilde{r}_n \sim \frac{n}{\ln (n)}$ as $n\rightarrow \infty$.
\end{theorem}


\begin{proof}

Consider the graph of $I_\gamma(n,r)$ against $r$.  The behaviour of $I_\gamma(n,r)$ is slightly involved, and so our proof uses the following strategy.
Let $t$ denote the ratio $r/n$, and so $0 \leq t \leq 1$, and let $\theta>0$ be a parameter that will take different values in the cases we consider (mostly we are concerned with
the cases where $0<\theta <1$ and $\theta >1$). 
For any $\delta \in (0,0.5)$ and any choice of $\theta$ we show that for $n$ sufficiently large, the graph of $I_\gamma$ has a gradient that is: 
 \begin{itemize}
 \item  greater than 1 for $t$ up to  $\frac{\theta}{\ln(n)}$, provided that  $\theta$ <1;
 \item  less than $-1$ for $t$ between $\frac{\theta}{\ln(n)}$  and $\delta$,  provided that $\theta >1$;
 \item less than $-1$ for $t$ between  $\delta$ and $1-\delta$;
 \item bounded by $C\sim 0.65$ for $t$  between $1-\delta$ and $1$.
 \end{itemize}
It follows that the (global) maximal value of $I_\gamma$ is given asymptotically (as $n$ grows)  by  $t \sim  \frac{1}{\ln (n)}$.  Note that the global maximal value cannot occur asymptotically (with $n$) at $t=1$ since the gradient of $I_\gamma$ is less or equal to $-1$ for $t$ over an interval of length (asymptotically with $n$)  at least $0.5$,  and the gradient is then bounded above by $C \sim 0.65$ for the remaining interval (i.e.  between $1-\delta$ and $1$) which has length less than $0.5$ (recall $\delta \in (0, 0.5)$).

Next we differentiate $I_{\gamma}(n, r)$ with respect to $r$. Writing
$$I_{\gamma}(n, r)=\ln\left(\frac{\gamma(n-r+1)}{\gamma\left(\frac{n}{r}\right)^r}\right),$$
and then replacing $\gamma(n-r+1)$ with $\frac{1}{\sqrt{2}}\left(\frac{2}{e}\right)^{n-r+1}(n-r+1)^{n-r}$ and $\gamma\left(\frac{n}{r}\right)$ with $\frac{1}{\sqrt{2}}\left(\frac{2}{e}\right)^{\frac{n}{r}}\left(\frac{n}{r}\right)^{\frac{n}{r}-1}$ and simplifying, we get
$$I_{\gamma}(n, r)=(1-r)\ln\left(\frac{\sqrt{2}}{e}\right)+(n-r)\ln\left(\frac{r(n-r+1)}{n}\right).$$
Differentiating $I_{\gamma}(n, r)$ with respect to $r$ gives:

\begin{equation}
\label{Igam}
\frac{d(I_\gamma(n,r))}{dr}= y(r)-z(r),
\end{equation}
where 
$$y(r)=\ln\left(\frac{e}{\sqrt{2}}\cdot\frac{n}{r(n+1-r)}\right) \mbox{ and } z(r) = \frac{(r-n)(n+1-2r)}{r(n+1-r)}.$$
Thus $ \frac{d(I_\gamma(n,r))}{dr}=0$ precisely  at values of $r$  for which   $y(r)-z(r)=0$. Note here that for $I_\gamma(n,r)$, the value $r$  can take any real value, not just integer values.
Let $t=t_n = r/n$. We may assume that $0 \leq t_n \leq 1$ for all $n$.  We will show that any value of $t_n$ that maximizes 
 $I_\gamma$  satisfies the asymptotic relationship $t_n \sim 1/\ln(n)$ (in other words, $\tilde{r}_n \sim n/\ln (n)$). 
Notice that if we let  $C=  \ln\left(\frac{e}{\sqrt{2}}\right)$ then we can write:
\begin{equation}
\label{ya}
y(r) = C -\ln( t) - \ln( n) - \ln\left(1+\frac{1}{n} -t\right).
\end{equation}
In addition, 
\begin{equation}
\label{za}
z(r) = \left(1-\frac{1}{t}\right) \cdot \frac{(1+ \frac{1}{n} -2t)}{(1+ \frac{1}{n} - t)}.
\end{equation}

We apply these equalities to firstly establish the following claims (which show that $\tilde{r}_n = o(n)$).  Suppose that  $\delta \in (0, 0.5)$.  We claim that:
\begin{itemize}
\item[(i)] If  $t \in [\delta, 1-\delta]$,  then
$
\frac{dI_\gamma(n,r)}{dr}  \leq h(n, \delta),
$
where $h(n, \delta)$ does not depend on $t$ and $h(n, \delta) <-1$ for all $n$ sufficiently large.
\item[(ii)] 
If $t \in [1-\delta, 1]$ and $n\geq 1$,  then 
$
\frac{dI_\gamma(n,r)}{dr}  \leq  K_\delta,
$
for a constant $K_\delta$ that converges to $C$ as $\delta \rightarrow 0$.
\end{itemize}

\bigskip

\noindent To establish Claim (i), Eqn~(\ref{ya}) implies that  $y(r) \leq C-\ln(\delta) - \ln(n)-\ln(\delta+ \frac{1}{n})$ and from  Eqn.~(\ref{za}) 
with $t \in [\delta, 1-\delta]$ we have $|z(r)| \leq \left(\frac{1}{\delta} -1 \right)\cdot \left| \frac{1+ \frac{1}{n} -2t}{1+ \frac{1}{n} - t}\right|$,
the second factor of which satisfies the inequality:
\begin{equation}
\label{esb}
 \left|\frac{1+ \frac{1}{n} -2t}{1+ \frac{1}{n} - t}\right| \leq \max\left\{1, \frac{|-1 + 2\delta +\frac{1}{n}|}{\delta + \frac{1}{n}}\right\}.
 \end{equation}
Thus,  $|z(r)| <  \left(\frac{1}{\delta} -1\right) a(n,\delta)$, where $a(n,\delta)$ is the bound on the right of Inequality~(\ref{esb}),
and so 
\begin{equation}
\label{esb2}
y(r)-z(r) \leq C-\ln(\delta) - \ln (n) - \ln\left(\delta+ \frac{1}{n}\right) + \left(\frac{1}{\delta} -1\right)a(n,\delta).
\end{equation}
If we now let $h(n, \delta)$ denote the term on the (entire) right-hand side of Inequality~(\ref{esb2})
then $h(n,\delta) \rightarrow -\infty$ as $n \rightarrow \infty$, which together with Eqn.~(\ref{Igam}) establishes
Claims (i).

To establish Claim (ii) note that  when $t \in [1-\delta, 1]$ we have $y(r) \leq C - \ln (1-\delta)$ and the right-hand-side converges to $C$ as $\delta \rightarrow 0$. Also,
 $-z(r) =\left (\frac{1}{t}-1\right) \cdot \frac{(1+ \frac{1}{n} -2t)}{(1+ \frac{1}{n} - t)}$ is 
 less or equal to zero for any value of $\delta < \frac{1}{2}$ once $n$ is sufficiently large.
This establishes Claim (ii). 

\bigskip

We next establish the following two claims:
\begin{itemize}
\item[(iii)] 
If $t \in [0, \frac{\theta}{\ln(n)}]$ and if  $\theta<1$, then
$
\frac{d I_\gamma(n,r)}{dr} \geq  h'(n, \theta),
$
where $h'(n,\theta)$ does not depend on $t$, and $h'(n, \theta) >1 $ for all $n$ sufficiently large.

\item[(iv)] If  $t \in [\frac{\theta}{\ln(n)}, \delta]$ and if  $\theta>1$ and $0<\delta< \frac{1}{\theta}$, then
$
\frac{dI_\gamma(n,r)}{dr} \leq  h''(n, \theta),
$
where $h''(n,\theta)$ does not depend on $t$, and $h''(n, \theta) <-1$  for all $n$ sufficiently large. 
\end{itemize}

\bigskip

To establish Claim (iii) observe that  $-\ln (t) >0$  (since $t<1$) and $-\ln\left(1+\frac{1}{n} -t\right)  \geq - \ln(2)$.  Thus,
 \begin{equation}
\label{yb}
y(r) \geq C' - \ln (n),
\end{equation}
where $C' = C-\ln (2)$. 
Moreover, since $0 \leq t \leq \frac{\theta}{\ln(n)}$  and since $\theta<1$ the second factor in the expression for $z(r)$ namely, $\frac{1+ \frac{1}{n} -2t}{1+ \frac{1}{n} - t}$ is bounded above by $1- \epsilon(n)$, where
$\epsilon(n)$ is a function only of $n$ that converges to zero as $n$ grows.  Thus we can write
\begin{equation}
\label{zb}
-z(r) \geq  \left(\frac{\ln(n)}{\theta}  - 1\right) (1- \epsilon(n)).
\end{equation}
 Combining Inequalities ~(\ref{yb}) and ~(\ref{zb}) gives
$\frac{dI_\gamma(n,r)}{dr}= y(r)-z(r)  \geq  h'(n, \theta)$,
where $$h'(n,\theta) = C' -\ln(n)\left(1- \frac{1}{\theta}(1-\epsilon(n))\right) - (1-\epsilon(n)).$$
Now  $h'(n,\theta) \rightarrow \infty$ as $n \rightarrow \infty$ (since $1- \frac{1}{\theta}(1-\epsilon(n))<0$ for all $n$ sufficiently large),
establishing Claim (iii).  
 
To establish Claim (iv), note that 
$y(r) \leq C-\ln\left(\frac{\theta}{\ln(n)} \right) - \ln(n)-\ln\left(1+ \frac{1}{n} - \frac{1}{\theta}\right).$
Moreover,  $-z(r) \leq \left(\frac{1}{t}-1\right) \leq \left(\frac{\ln(n)}{\theta} -1\right),$
and so $$y(r)-z(r) \leq - \left(1-\frac{1}{\theta}\right)\ln (n)  +\ln\left(\frac{\ln(n)}{\theta} \right) + C-1.$$
If we  take  $h''(n,\theta)$ to be the term on the right-hand side of this last inequality, we see that $h''(n,\theta)$ tends to $-\infty$ as $n \rightarrow \infty$, 
since $\left(1-\frac{1}{\theta}\right) >0$, thereby establishing Claim (iv).

It now follows from Claims (i) -- (iv)  that $I_\gamma(n,r)$ attains its maximal value at a value (or values) that can be written  
 $r= c_n \cdot \frac{n}{\ln (n)}$ where $c_n$ that converges to 1 as $n \rightarrow \infty$. 
 This completes the proof. 
 \hfill$\Box$
\end{proof}

\subsection{Remarks and questions} 
\label{remo}

For $n=120$, Theorem~\ref{asythm} gives the value $\tilde{r}_n \approx 25$, which is close to the exact value of $r_n=24$.  Fig.~\ref{rApprox} shows that $\tilde{r}_n$ provides a reasonable approximation to $r_n$ except for deviations near the `jumps'. Nevertheless it may well be that $\tilde{r}_n$ and $r_n$ are asymptotically equivalent (i.e. $\frac{\tilde{r}_n}{r_n}$ converges to 1 as $n \rightarrow \infty$) and the main step in a proof would be to first show that $n-r_n$ and $r_n$ both tend to infinity as $n\rightarrow \infty$.

Also, we have observed that `jumps' from $r_n$ to a smaller value $r_{n+1}$ tend to occur at values of $n$ for which $\frac{n}{r_n}$ is slightly greater than some integer (say $k$) while $\frac{n+1}{r_{n+1}}$ is slighly smaller than $k+1$ (for example, for the jump at $n=3484$,  $\frac{n}{r_n} = 7.01$, while $\frac{n+1}{r_{n+1}} = 7.99$). 
In that case:
$$\frac{n}{r_n} \approx \frac{n+1}{r_{n+1}}-1,$$
which rearranges to give the following estimate of the magnitude of a `jump' when $r_n> r_{n+1}$: 
$$ r_n - r_{n+1} \approx \frac{r_n(r_{n+1}-1)}{n}.$$
This is a partly heuristic (non-rigorous) argument, nevertheless the  approximation provides a reasonable estimate of the jump sizes  for the values reported in this paper. For example, for the jump that occurs at $n=3484$ where $r_n= 497$, while $r_{n+1} = 436$, we have
$$r_n - r_{n+1} = 61 \mbox{ while } \frac{r_n(r_{n+1}-1)}{n} \approx 62.05.$$


\section{Discussion} \label{sec:4}
In this manuscript, we analysed which characters have the highest information content. One of our main results is that in an optimal character with $r_n$ character states, all these states have to appear roughly equally often, as such a character can only induce at most two block sizes (which can differ by 1 at most). If $r$ divides the number $n$ of taxa, every block has the same size, $\frac{n}{r}$.

Concerning the behavior of $r_n$, the optimal number of states in order to maximize $I_\pi$, we found that although it has a generally increasing, partially linear trend, jumps occur (i.e. there are values of $n$ for which $r_{n+1}<r_n$). We analysed the reasons for these jumps, namely the shape of $I_\pi$, which is a concatenation of convex segments. Moreover, we presented an approximation for $I_\pi$, for which $n/\tilde{r}_n \sim \ln(n)$. Note that this does not directly imply that $n/r_n$ also tends to infinity, and formally establishing such a result could be an interesting exercise for future work.  All our theoretical statements were underlined by explicit calculations for up to $n=10000$. In order to be able to perform exhaustive searches for such large values of $n$, we had to find a region on which we can restrict the search. This, too, was done with the help of our approximation.
Some questions for future research have been raised (see Section~(\ref{remo})). More generally, determining the location of jumps as well as the location of block size changes (in terms of $n$) should lead to a deeper understanding of the most informative characters. 

Finally, as noted earlier, given any binary tree $T$ (involving {\em any} number of leaves) just  four characters (on a large enough number of states) suffice to ensure that $T$ is the only tree on which those four characters are convex (\cite{foursuffice}, \cite{bor06}).   A natural question is whether these four characters are of the `maximally informative' form as described in this paper. It turns out that for certain trees they divide up the leaf set $[n]$ quite differently. In particular, for a caterpillar tree, two of the characters described in \cite{foursuffice} 
partition the leaf set into (roughly) $n/2$ blocks of size $2$ while the other two characters partition the leaf set into one block of size (roughly) $n/2$ while the remaining leaf blocks are of size 1.

\begin{acknowledgements}
We thank the two anonymous reviewers for several helpful comments on an earlier version of this paper. 
I.D. and E.K. thank the International office at the University of Greifswald and the German Academic Exchange Service (DAAD) for the support through the mobility program PROMOS (travel scholarship). We also thank the (former) Allan Wilson Centre for supporting this research.
\end{acknowledgements}

\bibliography{Refs}
\bibliographystyle{spbasic}

     \newpage

\section{Appendix}

Proof of Lemma~\ref{lem:equalblocksizes}

We first consider the case $m=2$. In this case, we have $b(m+s)=b(2+s)$ and $b(m)=b(2)=1$ as well as $b(m+s-1)=b(s+1)$ and $b(m+1)=b(3)=1$.  In total, we have $b(m+s) \cdot b(m) = b(s+2) > b(s+1) = b(m+1)\cdot b(m+s-1)$, which is true for all $s\geq 2$.

We now consider the case $m\geq 3$. 
As $s\geq 2$, we have: \begin{eqnarray*} 2m + 2s &>& 2m +2 \\ \Rightarrow 2(m+s)-5&> &2m-3 \\ \Rightarrow (2(m+s)-5) \cdot (2(m+s)-7)!! \cdot (2m-5)!! &>& (2m-3) \cdot(2(m+s)-7)!! \cdot (2m-5)!! \\ \Rightarrow (2(m+s)-5)!! \cdot (2m-5)!! &>& (2m-3)!! \cdot (2(m+s)-7)!! \\ \Rightarrow (2(m+s)-5)!! \cdot (2m-5)!! &>& (2(m+1)-5)!! \cdot (2(m+s-1)-5)!! \\ \Rightarrow b(m+s) \cdot b(m) &>& b(m+1)\cdot b(m+s-1).\end{eqnarray*}
The last line uses the fact that $b(m)=(2m-5)!!$ for all $m\geq 3$. This completes the proof. \qed

Note that Lemma \ref{lem:equalblocksizes} is only stated for $m \geq 2$. If $m=1$, the lemma only holds for $s\geq 3$. To see this, consider the case $m=1$ and $s=2$. Then,  $b(m+s)\cdot b(m)= b(1+2)\cdot b(1) = b(1+1)\cdot b(1+2-1)= b(m+1)\cdot b(m+s-1)$, as $b(1)=b(2)=b(3)=1$. Therefore the strict inequality stated in the lemma no longer holds.

\end{document}